\documentclass[11pt]{article}
\usepackage{amsmath,amsthm,amsfonts,amssymb,url}
\newtheorem{theorem}{Theorem}
\newtheorem{lemma}[theorem]{Lemma}

\newtheorem{example}{Example}
\newcommand{\GL}{\ensuremath{\mathbb{GL}}}
\newcommand{\zed}{\ensuremath{\mathbb{Z}}}
\newcommand{\eff}{\ensuremath{\mathbb{F}}}
\newcommand{\conj}{\ensuremath{\mathsf{conj}}}
\newcommand{\cent}{\ensuremath{\mathsf{C}}}

\title{A Brief Retrospective Look at the Cayley-Purser Public-key 
Cryptosystem, 19 Years Later}
\author{Douglas~R.~Stinson\thanks{The author's research is supported by  
NSERC discovery grant RGPIN-03882.}\\
David R.\ Cheriton School of Computer Science\\ University of Waterloo\\
Waterloo, Ontario N2L 3G1, Canada}

\begin{document}
\maketitle

\begin{abstract}
The purpose of this paper is to describe and analyze the Cayley-Purser algorithm, which is
a public-key cryptosystem proposed by Flannery in 1999.
I will present two attacks on it, one of which is apparently new. 
I will also examine a variant of the Cayley-Purser
algorithm that was patented by Slavin in 2008, and show that it is also insecure.
\end{abstract}

\section{Introduction}

When she was only 16 years of age, Sarah Flannery won the {\it EU Young Scientist of the Year Award} for 1999.
Her project consisted of a proposal of a public-key cryptosystem based on $2$ by $2$
matrices with entries from  $\zed_n$, where $n$ is the product of
two distinct primes $p$ and $q$. The cryptosystem she proposed was
named the {\it Cayley-Purser algorithm}.\footnote{The cryptosystem was named after the mathematicians Arthur Cayley and Michael Purser. Flannery \cite{Flan} 
states that the Cayley-Purser
algorithm was based in part on ideas in an unpublished paper by Michael Purser.} 

Because this algorithm
was faster than the famous {\it RSA public-key cryptosystem}, it 
garnered an incredible amount of press coverage in early 1999;
see, for example, the BBC News article \cite{BBC} published on January 13, 1999. 
However, at the time of this press coverage, the algorithm had not undergone 
any kind of serious peer review.
Unfortunately, the Cayley-Purser algorithm was shown to be insecure
later in 1999, e.g., as reported by Bruce Schneier \cite{Schneier} in December, 1999.

Ms Flannery later wrote an interesting book, entitled
{\it In  Code: A Mathematical Journey} \cite{Flan},
which recounts her experiences relating to her work on the algorithm.
The technical description and the analysis of the Cayley-Purser algorithm,
as well as an attack on it, 
are found in \cite[Appendix A]{Flan}.

In this paper, I will describe the Cayley-Purser algorithm
and two attacks on it, one of which is apparently new. I will also comment
a bit on the underlying mathematical theory. Finally, I will examine a variant of the Cayley-Purser
algorithm, which was patented in 2008 by Slavin, and show that it is also insecure.

\section{The Cayley-Purser Algorithm}

In this section, we describe the Cayley-Purser algorithm, which is
presented in \cite[pp.\ 274--277]{Flan}.
Note that all material in this section is paraphrased from \cite{Flan}.

\bigskip

\noindent {\bf Setup}: Let $n = pq$, where $p$ and $q$ are large distinct primes.
(We assume that it is infeasible to factor $n$.) 
$\GL(2,n)$ denotes the $2$ by $2$ invertible matrices 
with entries from $\zed_n$.  Let $A,C \in \GL(2,n)$
be chosen such that $AC \neq CA$. Define $B = C^{-1}A^{-1}C$.
Then choose a secret, random positive integer $r$ and let $G = C^r$.

\bigskip

\noindent The {\em  public key} consists of $A, B, G, n$.

\bigskip

\noindent The {\em private key} consists of $C,p,q$. 

\bigskip

\noindent {\bf Encryption}: Let $X \in \GL(2,n)$ be the plaintext to be
encrypted. The following computations are performed:
\begin{enumerate}
\item choose a secret, random positive integer $s$
\item compute $D = G^s$
\item compute $E = D^{-1}AD$
\item compute $K = D^{-1}BD$
\item compute $Y = KXK$
\item the ciphertext is $(E,Y)$.
\end{enumerate}

\bigskip

\noindent {\bf Decryption}: Let $(E,Y) \in \GL(2,n) \times \GL(2,n)$ 
be the ciphertext to be
decrypted. The following computations are performed:
\begin{enumerate}
\item compute $L = C^{-1}EC$ (note: $L = K^{-1}$)
\item compute $X = LYL$
\end{enumerate}
Observe that the factorization $n=pq$ is not needed in order to decrypt ciphertexts; 
the matrix $C$ is all that is required.

The correctness of the decryption process is easy to show.

\begin{theorem} \cite{Flan}
If the ciphertext $(E,Y)$ is an encryption of the plaintext $X$, then the decryption
of $(E,Y)$ yields $X$.
\end{theorem}
\begin{proof}
First we show that $L  = K^{-1}$:
\[
\begin{array}{rcll}
LK & = & (C^{-1}EC)(D^{-1}BD) &\text{substituting for $L$ and $K$ }\\
 & = & C^{-1} (D^{-1}AD)  CD^{-1}BD & \text{substituting for $E$  }\\
 & = & D^{-1}C^{-1} A CD D^{-1}BD & \text{because $C$ and $D$ commute} \\
 & = & D^{-1}C^{-1} A CBD & \text{cancelling $D D^{-1}$} \\
 & = & D^{-1}B^{-1}BD & \text{because  $B^{-1} = C^{-1} A C$} \\
 & = & I.
\end{array}
\]
Then it is easy to verify that
\[ LYL  = K^{-1} Y K^{-1} = X.\]
\end{proof}

\section{Two Attacks}
\label{attacks.sec}

\noindent The basis of the two attacks we will describe is the observation 
from \cite[p.\ 290]{Flan} 
that any scalar multiple $\mu C$  can be used in place of $C$ in the decryption
process. This is easy to see, because
\begin{equation}
\label{scalar.eq}
 (\mu C)^{-1}E (\mu C) = C^{-1}EC.
\end{equation}
Therefore, using $\mu C$ in step 1 of the decryption process still results in the
correct value of $L$ being computed.

Thus, it is sufficient for an attacker to compute $C$ up to a scalar multiple.
This will allow any ciphertext to be decrypted, since the factorization $n=pq$
is not required in order to be able to decrypt ciphertexts.

\subsection{Linear Algebra Attack}

The attack described in this section is very simple but apparently new. 
It turns out to be straightforward
to construct the private key $C$ (or a scalar multiple $\mu C$) directly from
the public key by solving a certain
system of linear equations in $\zed_n$. 
We make use of the following two equations involving $C$:
\begin{equation}
\label{eq1} CB = A^{-1}C
\end{equation}
and
\begin{equation}
\label{eq2} CG = GC
\end{equation}
Note that (\ref{eq1}) follows from the formula $B = C^{-1}A^{-1}C$.
It is also clear that (\ref{eq2}) holds because $G$ is a power of $C$ and hence
$G$ and $C$ commute.

We observe that (\ref{eq1}) and (\ref{eq2}) are sufficient to compute 
$C$, up to a scalar multiple, by solving a system of linear equations
in $\zed_n$. In these equations, $A, B$ and $G$ are known matrices
and we are trying to determine $C$. Let 
\begin{equation}
\label{C.eq}
 C =
\left(
\begin{array}{cc}
a & b\\
c & d
\end{array}
\right),
\end{equation} where $a,b,c,d \in \zed_n$.
Then (\ref{eq1}) and (\ref{eq2}) each yield four homogeneous linear equations
(in $\zed_n$) in the four unknowns $a,b,c,d$. 
The solution space of (\ref{eq1}) is a $2$-dimensional subspace of $(\zed_n)^4$,
as is the solution space of (\ref{eq2}). However, when we solve all eight equations
simultaneously, we get precisely the scalar multiples of $C$ (i.e., the solution space is a
$1$-dimensional subspace of $(\zed_n)^4$).

We will justify the statements made above in the next section.
For now,  we illustrate the attack with a toy example.

\begin{example}
\label{ex1}
Suppose $p = 193$ and $q= 149$, so $n = 28757$. Suppose we define
\[ A =
\left(
\begin{array}{cc}
16807 & 19399\\
7483 & 18143
\end{array}
\right)
\]
and 
 \[ C =
\left(
\begin{array}{cc}
2910 & 1657\\
5341 & 24803
\end{array}
\right).
\] 
Then 
\[ B =
\left(
\begin{array}{cc}
11947 & 1712\\
4630 & 14946
\end{array}
\right).
\] 
Finally, suppose $G = C^7$; then 
\[ G =
\left(
\begin{array}{cc}
1438 & 1433 \\ 20759 & 24068
\end{array}
\right).
\] 
The system of linear equation to be solved is
\[ \left(
\begin{array}{cccc}
24034 & 4630 & 19287 & 0 \\
1712 & 27033 & 0 & 19287\\
9570 & 0 & 1724 & 4630 \\
0 & 9570 & 1712 & 4723 \\
0 & 20759 & 27324 & 0 \\
1433 & 22630 & 0 & 27324 \\
7998 & 0 & 6127 & 20759 \\
0 & 7998 & 1433 & 0
\end{array}
\right)
\left(
\begin{array}{c}
a \\ b \\ c \\ d
\end{array}
\right) =
\left(
\begin{array}{c}
0 \\ 0 \\ 0 \\ 0
\end{array}
\right).
\]
The solution to this system is
\[ (a,b,c,d) = \mu (28365, 13928, 25231, 28756),\] $\mu \in \zed_n$.
It is straightforward to
verify that this solution space indeed consists of all the scalar multiples of $C$.
\end{example}

\subsection{Cayley-Hamilton Attack}

The other attack I will present is the original attack presented in \cite[pp.\ 290--292]{Flan}. 
It is in fact even more efficient than the attack we just described above.
We summarize it briefly now.

The Cayley-Hamilton theorem states that every square matrix $A$ over a commutative ring satisfies its own characteristic polynomial. The {\em characteristic polynomial}
of $A$ is the polynomial $\det(x I_{n}-A)$ in the indeterminate $x$, 
where $A$ is an $n$ by $n$ matrix
and $I_n$ is the $n$ by $n$ identity matrix.
When $n = 2$, the characteristic polynomial is quadratic. In this case, 
as noted in \cite[p.\ 291]{Flan}, it follows that 
any power of $A$ can can be expressed as a linear combination of $A$ and $I_2$.

Recall that $G$ is a power of $C$ and hence $C$ is  also a power of $G$. So the 
unknown matrix $C$ can be expressed in the form $C = \alpha I_2 + \beta G$, for scalars
$\alpha$ and $\beta$. Since we only have to determine $C$ up to a scalar multiple, 
we can WLOG take $\beta = 1$, and write $C = \alpha I_2 +  G$ (we are ignoring here the unlikely possibility that $\beta = 0$). Suppose we substitute this
expression for $C$ into (\ref{eq1}). Then we obtain
\[ (\alpha I_2 +  G)B = A^{-1}(\alpha I_2 +  G) .\]
Rearranging this, we have
\[ \alpha (B - A^{-1}) = A^{-1}G - GB.\]
If we compute the two matrices $B - A^{-1}$ and $A^{-1}G - GB$, we can compare
any two corresponding nonzero entries of these two matrices to determine $\alpha$. 

\begin{example}
We use the same parameters as in Example \ref{ex1}.
First we compute 
\[
B - A^{-1} = 
\left(
\begin{array}{cc}
24034 & 20999 \\ 14200 & 4723
\end{array}
\right).
\] 
and 
\[
A^{-1}G - GB = 
\left(
\begin{array}{cc}
17977 & 4614 \\ 25427 & 10780
\end{array}
\right).
\] 
From this, we see that
\[  28534 (B - A^{-1}) = A^{-1}G - GB, \]
so $\alpha = 28534$.
Hence, 
\[ 28534 I_2 +  G  = \left(
\begin{array}{cc}
1215 & 1433 \\ 20759 & 23845
\end{array}
\right)\]
should be a multiple of $C$. In fact, it can be verified that 
\[ \left(
\begin{array}{cc}
1215 & 1433 \\ 20759 & 23845
\end{array}
\right) = 5485 C.
\]
\end{example}

\section{Discussion and Comments}

When the Cayley-Purser algorithm was proposed, there was some mathematical 
analysis provided to justify its security against certain types of attacks
\cite[pp.\ 277--283]{Flan}.
There are some interesting mathematical points related to this that 
I would like to discuss in this section. I will also look briefly at the 
efficiency of encryption and decryption.

\subsection{Security Analysis from \cite{Flan}}

The main possible attack discussed in \cite[pp.\ 277--283]{Flan} involves trying to use
(\ref{eq1}) to compute $C$ (or a scalar multiple of $C$). The argument given
is that the number of solutions (for $C$) to (\ref{eq1}) is so large that it would
be infeasible to distinguish the real value of $C$ from the extra ``bad'' solutions 
to (\ref{eq1}). It is noted that the number of solutions for $C$ is equal to 
$|\cent _{\GL(2,n)}(A^{-1})|$, where $\cent _{\GL(2,n)}(A^{-1})$ 
denotes the centralizer of
$A^{-1}$, i.e., the set of matrices in $\GL(2,n)$ that commute with $A^{-1}$.
(The actual set of solutions to (\ref{eq1}) is a coset of $\cent _{\GL(2,n)}(A^{-1})$.)

Then, a lower bound on $|\cent _{\GL(2,n)}(A^{-1})|$ is obtained from the observation
that every power of $A^{-1}$ (or, equivalently, every power of $A$) is an element of the set  
$\cent _{\GL(2,n)}(A^{-1})$. Hence, $|\cent _{\GL(2,n)}(A^{-1})| \geq \mathit{ord}(A)$. Then, an analysis
of the number of group elements of all possible orders is done, and it is shown that
most group elements have order that is close to $n^2$. Since there are only $n$
scalar multiples of the correct $C$, there are many ``bad'' solutions remaining.

The above-described analysis is correct. But, more precisely, it turns out that it is 
fairly straightforward to determine the exact number of solutions to (\ref{eq1}) using
some standard group theoretic arguments. Note also that the solution space of (\ref{eq1}) 
or (\ref{eq2}) contains tuples $(a,b,c,d)$ where the corresponding matrices
(\ref{C.eq}) turn out not be invertible. 

We need some definitions to get started.
For now, we confine our attention to $\GL(2,q)$ for a prime $q$. The following results
are found in various standard algebra textbooks, such as Dummit and Foote \cite{DM}.
Details of these calculations are presented in 
Mathewson \cite{Mathewson}.

Two matrices $A$ and $B$ are {\em similar} if $B = C^{-1}AC$ for some matrix $C$.
(Thus, if (\ref{eq1}) holds, then $A^{-1}$ and $B$ are similar.)
Similarity is an equivalence relation and the equivalence classes under similarity
are known
as {\em conjugacy classes}. The conjugacy class containing $A$ is denoted
by $\conj (A)$. It follows from the orbit-stabilizer theorem 
that 
\begin{equation}
\label{eq3}
|\GL(2,q)| = |\cent _{\GL(2,n)}(A)| \cdot |\conj (A)|
\end{equation} 
for any $A \in \GL(2,q)$.
Further, it is well-known that
\begin{equation}
\label{eq4} |\GL(2,q)| = (q^2 -1)(q^2-q).
\end{equation}

Now, it is fairly easy to determine the various conjugacy classes by
using the fact that any conjugacy class contains a unique matrix in 
\emph{rational canonical form}. The rational canonical forms in $\GL(2,q)$
have the following possible structures:
\begin{description}
\item [case (1)] \[ \left(
\begin{array}{cc}
a & 0\\
0 & a
\end{array}
\right) .\]
\item [case (2)]
\[ \left(
\begin{array}{cc}
0 & b\\
1 & c
\end{array}
\right).\]
\end{description}
Case 2 further subdivides into three subcases:
\begin{description}
\item [case (2a)] $b^2 + 4a$ is not a perfect square in $\zed_q$,
\item [case (2b)] $b^2 + 4a = 0$ in $\zed_q$, and
\item [case (2c)] $b^2 + 4a$ is a nonzero perfect square in $\zed_q$.
\end{description}
 Further, for a given matrix expressed in
rational canonical form, it is relatively straightforward to determine
$|\cent _{\GL(2,q)}(A)|$. Then $|C_A|$ can also be determined, from (\ref{eq3})
and (\ref{eq4}).
Table \ref{tab1} lists the number of conjugacy classes of all possible sizes
(note that these results are all given in \cite{Mathewson}).

\begin{table}
\caption{The number of conjugacy classes in $\GL(2,q)$ of all possible sizes}
\label{tab1}
\begin{center}
\begin{tabular}{c|c|c}
Case & Size of conjugacy class & Number of conjugacy classes \\ \hline\hline
case (1) & $1$ & $q-1$ \\\hline
case (2a) & $q^2 - q$ & ${\frac{q^2-q}{2}}$ \\\hline
case (2b) & $q^2-1$ & $q-1$ \\\hline
case (2c) & $q^2+q$ & ${\frac{(q-1)(q-2)}{2}}$
\end{tabular}
\end{center}
\end{table}

The Cayley-Purser algorithm lives in $\zed_n$. So the relevant
sizes of conjugacy classes
would be obtained by working modulo $p$ and modulo $q$, and  then applying
the Chinese remainder theorem to derive the sizes of the conjugacy classes 
in $\GL(2,n)$. The vast majority of these conjugacy classes in $\GL(2,n)$ 
have size very close to $n^2$, which indicates that the solution to (\ref{eq1})
will be a two-dimensional subspace of $(\zed_n)^4$.

\bigskip

The second possible attack considered in \cite{Flan} involves trying to determine
the private key $C$ from the public key $G$. It is known that $G = C^r$, where $r$ is
secret. However, $r$ might be chosen from a small range of values (in \cite{Flan}, 
$r \leq 50$). So we might consider  trying various values of $r$ until the equation
 $G = C^r$ can be solved. However, even if $r$ is known, it is not easy to solve this equation.
 For example, consider the special case where $r = 2$ and $G$ is a scalar multiple of the identity.
 Solving for $C$ is then equivalent in difficulty to extracting square roots in $\zed_n$,
 which is equivalent to factoring $n$. So this particular attack will not succeed.

Of course, these two analyses are not sufficient to establish the security of the 
Cayley-Purser algorithm. As we saw in the previous section, an attack that utilizes
all the public information allows $C$ to be computed up to a scalar multiple, which breaks
the cryptosystem.

\subsection{Efficiency of Encryption and Decryption}

We also have a few comments about the efficiency of encryption and decryption
in the Cayley-Purser algorithm. One of the attractive features of 
the Cayley-Purser algorithm is its speed relative to RSA. It is reported 
in \cite[pp.\ 284--289]{Flan} that Cayley-Purser encryption and decryption is
roughly 20--30 times faster than the comparable RSA operations.

Clearly Cayley-Purser decryption is much faster than RSA decryption, because
Cayley-Purser decryption just requires a few fast matrix operations, whereas
RSA decryption uses an exponentiation modulo $n$. On the other hand, Cayley-Purser encryption involves exponentiating the matrix $G$, which is an expensive operation.
However, there is a trick that can be used to speed up encryption. A careful
reading of the Mathematica code that is provided in \cite{Flan} shows that
step 2 of the encryption method is implemented by computing a linear combination
of $G$ and the identity. Using the Cayley-Hamilton theorem, it can easily be shown that this is a quicker 
way of obtaining a matrix $D$ that is actually a power of $G$. With this modification
to the encryption algorithm, no matrix exponentiations are required to encrypt a plaintext.

\section{A Variation due to Slavin}

In this section, I discuss a variation of the Cayley-Purser algorithm due to Slavin \cite{Slav}.
I am not aware of any analysis of this algorithm in the cryptographic literature.
However, it is not difficult to see that it is also insecure.

The following description is from the 2008  U.S.\ patent \cite{Slav}. It is clear that
this  cryptosystem is similar to the Cayley-Purser algorithm in many respects;
however, several of the equations have been modified. 

\bigskip

\noindent {\bf Setup}: Let $n = pq$, where $p$ and $q$ are distinct primes. 
Let $A,C \in \GL(2,n)$
be chosen such that $AC \neq CA$. Define $B = CAC$.
Then choose a secret, random positive integer $r$ and let $G = C^r$.

\bigskip

\noindent The {\em  public key} consists of $A, B, G, n$.

\bigskip

\noindent The {\em private key} consists of $C,p,q$. 

\bigskip

\noindent {\bf Encryption}: Let $X$ be the plaintext to be
encrypted. The following computations are performed:
\begin{enumerate}
\item choose a secret, random positive integer $s$
\item compute $D = G^s$
\item compute $E = DAD$
\item compute $K = DBD$
\item let $Y = e_K(X)$ under some secret-key cryptosystem such as AES.
\item the ciphertext is $(E,Y)$.
\end{enumerate}
{\bf Remark:} The value $K$ is used as a key in a secret-key cryptosystem. This is
different from the Cayley-Purser algorithm, but it does not affect the security of this cryptosystem.

\bigskip

\noindent {\bf Decryption}: Let $(E,Y)$ 
be the ciphertext to be
decrypted. The following computations are performed:
\begin{enumerate}
\item compute $L = CEC$
\item compute $X = d_L(Y)$
\end{enumerate}

Using the fact that $C$ and $D$ commute, it is not difficult to verify that $CEC = DBD$ 
and therefore $L = K$; hence, decryption will succeed.

\subsection{The Attack}

Our attack is based on the  following  observation from \cite{Slav}.

\begin{lemma}
Define $M = BGB^{-1}$ and $N = AGA^{-1}$. Then $M = CNC^{-1}$.
\end{lemma}

\begin{proof}
We compute as follows:
\[
\begin{array}{rcll}
CNC^{-1} & = & C(AGA^{-1}) C^{-1} &\text{substituting for $N$ }\\
 & = & CAC C^{-1} G C C^{-1} A^{-1} C^{-1} & \text{inserting $C C^{-1}$ twice  }\\
 & = & B C^{-1} G C B^{-1} & \text{because $B = CAC$ } \\
 & = & B G C^{-1}  C B^{-1} & \text{because $G$ and $C$ commute} \\
 & = & B G B^{-1} & \text{cancelling  $C^{-1}  C$} \\
 & = & M.
\end{array}
\]
\end{proof}

We now describe our attack on Slavin's cryptosystem.
First, note that $N$ and $M$ can both be computed from public information. 
Using the two equations $M = CNC^{-1}$ and $GC = CG$, we can carry out either of the attacks 
described in Section \ref{attacks.sec} to compute a scalar multiple of the unknown matrix $C$, say $C'$.
Thus $C = \mu C'$ for some unknown value $\mu \in {\zed_n}^*$.

Slavin \cite{Slav} argues that, unlike the situation in the Cayley-Purser algorithm,
it is not sufficient to compute a scalar multiple of $C$.
In the Cayley-Purser algorithm, equation (\ref{scalar.eq}) 
allows $K^{-1}$ to be computed by an attacker
using any scalar multiple of $C$. On the other hand, in Slavin's cryptosystem,
the ``key'' $K = CEC$. If we replace $C$ by a scalar multiple, then the attacker
doesn't obtain the correct value of $K$.

However, an attacker can compute $K$ by a slightly different approach.
Consider the equation $B = CAC$. We can rewrite this as
$B = \mu^2 C'AC'$. From this, it is a simple matter to compute $\mu^2$.
Computing $\mu$ is infeasible unless the factorization of
$n$ is known; however, it turns out that we do not need to compute $\mu$.

Finally, consider the equation $K = CEC$. We can rewrite this as
$K = \mu^2 C'EC'$. Since $C', E$ and $\mu^2$ are known, the attacker can compute $K$
and use it to decrypt the ciphertext $Y$.

Thus, the steps in the attack are summarized as follows:
\begin{enumerate}
\item Compute $M$ and $N$ from $A$, $B$ and $G$.
\item Compute $C'$, where $C = \mu C'$ for some unknown value $\mu$.
\item Use the equation $B = \mu^2 C'AC'$ to compute $\mu^2$.
\item Given a ciphertext $(E,Y)$, compute $K = \mu^2 C'EC'$.
\item Use $K$ to decrypt $Y$.
\end{enumerate}
Observe that steps 1--3 only involve the public key; they only need to be carried out once.
Steps 4--5 then allow the decryption of a specific ciphertext; they can be repeated
as often as desired, for various ciphertexts.

\begin{example}
\label{ex3}
Suppose $p = 223$ and $q= 173$, so $n = 38579$. Suppose we define
\[ A =
\left(
\begin{array}{cc}
16807 & 38390\\
17333 & 21788
\end{array}
\right)
\]
and 
 \[ C =
\left(
\begin{array}{cc}
10106 & 10420\\
27722 & 27626
\end{array}
\right).
\] 
Then 
\[ B =
\left(
\begin{array}{cc}
17590 & 36066\\
32833 & 33331
\end{array}
\right).
\] 
Finally, suppose $G = C^{11}$; then 
\[ G =
\left(
\begin{array}{cc}
11303 & 17971 \\ 
5315 & 18194
\end{array}
\right).
\] 
The attack begins by computing $M$ and $N$:
\[ M =  BGB^{-1} =
\left(
\begin{array}{cc}
18545 & 20365 \\ 
25987 & 10952
\end{array}
\right)
\] 
and 
\[ N = AGA^{-1} = 
\left(
\begin{array}{cc}
37716 & 5184 \\ 
18941 & 30360
\end{array}
\right).
\] 
Using the linear algebra attack,
the system of linear equation to be solved is
\[ \left(
\begin{array}{cccc}
19171 & 18941 & ,18214 & 0\\
5184 & 11815 & 0 & 18214\\
12592 & 0 & 26764 & 18941\\
0 & 12592 & 5184 & 19408\\
0 & 5315 & 20608 & 0\\
17971 & 6891 & 0 & 20608\\
33264 & 0 & 31688 & 5315\\
0 & 33264 & 17971 & 0
\end{array}
\right)
\left(
\begin{array}{c}
a \\ b \\ c \\ d
\end{array}
\right) =
\left(
\begin{array}{c}
0 \\ 0 \\ 0 \\ 0
\end{array}
\right).
\]
The solution to this system is
\[ (a,b,c,d) = \mu (12688 ,23061,22337,38578),\] 
$\mu \in \zed_n.$

Let 
\[C' = 
\left(
\begin{array}{cc}
12688 & 23061\\
22337 & 38578
\end{array}
\right).
\]
Then $C'$ is an unknown scalar multiple of $C$.
However, the attacker can compute 
\[ C' A C' = 
\left(
\begin{array}{cc}
27011 & 27739\\
26956 & 8680
\end{array}
\right)
\]
By comparing $B$ to $C'AC'$, it is easy to see that
$\mu^2 = 26098$.

Now suppose a plaintext is encrypted. 
First, $D = G^s$ is computed for a random exponent $s$.
Suppose that $D = G^{129}$; then
\[D = 
\left(
\begin{array}{cc}
18776 & 31218\\
20617 & 22838
\end{array}
\right).
\]
Then 
\[E = DAD = 
\left(
\begin{array}{cc}
33712 & 19745\\
30382 & 3658
\end{array}
\right)
\]
and
\[K = DBD = 
\left(
\begin{array}{cc}
33935 & 21771\\
36280 & 7314
\end{array}
\right)
\]

Given $E$, the attacker can compute 
\[ \mu^2 C' E C' = 
\left(
\begin{array}{cc}
33935 & 21771\\
36280 & 7314
\end{array}
\right),
\]
which yields the ``key'' $K$.
\end{example}

\section{Final Comments}

The Cayley-Purser algorithm was a huge news story in early 1999.
However, like many other ``broken'' cryptosystems, it has been forgotten
to a certain extent. I hope that this paper serves to highlight some
interesting mathematical techniques that can be used to analyze and
break this cryptosystem as well as the later, lesser-known variant 
that was patented by Slain in 2008.



\end{document}